\documentclass{llncs}

\usepackage[nocompress]{cite}
\usepackage{mathrsfs}
\usepackage{amsmath, amssymb}
\usepackage{amscd}
\usepackage{subfigure,color}
\usepackage{graphicx,graphics}
\usepackage{setspace}

\newcommand{\pig}[1]{\mathrm{int}(#1)}

\newcommand{\tr}[1]{\mathcal{T}_r#1}
\newcommand{\cliquemap}{\mathcal{K}}
\newcommand{\treerep}[1]{\tr{#1} = (T{#1},\cliquemap{#1})}

\newcommand{\pc}{\mathcal{C}}

\newcommand{\xt}[1]{\mathcal{T}_x#1}
\newcommand{\xtree}[1]{\xt{#1} = (T{#1},\phi{#1})}

\newcommand{\isomorphic}{\cong}

\newcommand{\comment}[1]{}

\newtheorem{observation}{Observation}

\begin{document}

\title{Unique Perfect Phylogeny Characterizations via Uniquely Representable Chordal Graphs}

\author{Rob Gysel}
\institute{Department of Computer Science\\University of California, Davis\\1 Shields Avenue, Davis CA 95616, USA\\ \email{rsgysel@ucdavis.edu}}

\maketitle
\begin{abstract}
The perfect phylogeny problem is a classic problem in computational biology, where we seek an unrooted phylogeny that is compatible with a set of qualitative characters.
Such a tree exists precisely when an intersection graph associated with the character set, called the partition intersection graph, can be triangulated using a restricted set of fill edges.
Semple and Steel used the partition intersection graph to characterize when a character set has a unique perfect phylogeny.
Bordewich, Huber, and Semple showed how to use the partition intersection graph to find a maximum compatible set of characters.
In this paper, we build on these results, characterizing when a unique perfect phylogeny exists for a subset of partial characters.
Our characterization is stated in terms of minimal triangulations of the partition intersection graph that are uniquely representable, also known as ur-chordal graphs.
Our characterization is motivated by the structure of ur-chordal graphs, and the fact that the block structure of minimal triangulations is mirrored in the graph that has been triangulated.
\end{abstract}

\section{Introduction}
An \emph{$X-$tree} is a pair $\xtree{}$ where $T$ is a tree and $\phi$ is a map from $X$ to the nodes of $T$, such that every node of $T$ with degree two or one is mapped to by $\phi$.
We will call the range of $\phi$ the \emph{labeled nodes} of $\xt{}$, and these nodes are \emph{labeled by} $\phi$.
The \emph{underlying tree} of $\xt{}$ is $T$.
An $X-$tree is \emph{free} if $\phi$ is a bijection to the leaves of $T$, and it is \emph{ternary} if every internal node of $T$ has degree three.
Given $A \subseteq X$, we will use $\xt{}(A)$ to denote the minimal subtree of $T$ containing the nodes $\phi(A)$.
Two subtrees $\xt{}(A)$ and $\xt{}(A')$ of $T$ \emph{intersect} if they have one or more nodes in common, and if $v$ is a common node of $\xt{}(A)$ and $\xt{}(A')$ we say that $\xt{}(A)$ and $\xt{}(A')$ \emph{intersect at} $v$.

A \emph{partial character} for $X$ is a partition $\chi = A_1 | A_2 | \ldots | A_r$ of a subset $X' \subseteq X$.
Each $A_i$ is called a \emph{cell} of $\chi$.
If $\xt{}(A)$ and $\xt{}(A')$ do not intersect for every pair of distinct cells $A$ and $A'$ of $\chi$, then $\xt{}$ \emph{displays} $\chi$.
A \emph{perfect phylogeny} for a set of partial characters $\pc$ is an $X-$tree $\xt{}$ that displays each character in $\pc$.
When $\pc$ has a perfect phylogeny, we also say that $\pc$ is \emph{compatible}.
The \emph{perfect phylogeny problem} (also called the \emph{character compatibility problem}) is to determine if a set of partial characters is compatible.

The perfect phylogeny problem reduces to a graph theoretical problem that we detail now.
Given a set of characters $\pc$, one can construct the \emph{partition intersection graph} $\pig{\pc}$ as follows. 
The vertex set of $\pig{\pc}$ is 
\begin{equation*}
	\{(A,\chi) \mid \chi \in \pc \mbox{ and } A \mbox{ is a cell of } \chi\} \enspace, 
\end{equation*}	
and there is an edge between two vertices $(A,\chi)$ and $(A',\chi')$ if and only if $A$ and $A'$ have non-empty intersection.
For a vertex $(A,\chi)$ of $\pig{\pc}$, $A$ is the \emph{cell} of $(A,\chi)$ and $\chi$ is the \emph{character} of $(A,\chi)$.
Observe that if $\chi = A_1 | A_2 | \ldots | A_r$ is a partial character, then every pair of distinct vertices $(A,\chi)$ and $(A',\chi)$ are non-adjacent in $\pig{\pc}$.

A graph is \emph{chordal} if every cycle of length four or more has a \emph{chord}, that is, an edge between vertices of the cycle that do not appear consecutively in the cycle.
In general $\pig{\pc}$ is not a chordal graph, and we are interested in adding edges to $\pig{\pc}$ to obtain a chordal supergraph $H$ of $\pig{\pc}$ that is called a \emph{triangulation} of $\pig{\pc}$.
The added edges are called \emph{fill edges}.
If no subset of the fill edges yields a triangulation of $\pig{\pc}$, it is a \emph{minimal triangulation} of $\pig{\pc}$.
When each fill edge is of the form $(A,\chi)(A',\chi')$ and $\chi \neq \chi'$, the resulting triangulation is a \emph{proper triangulation} of $\pig{\pc}$.
The following classic result reduces the question of determining compatibility to finding proper triangulations of the partition intersection graph.
It was originally phrased in terms of proper triangulations, but from the definitions it follows that $\pig{\pc}$ has a proper triangulation if and only if it has a proper minimal triangulation.

\begin{theorem}\label{thm_buneman}\cite{B74,M83,S92}
	Let $\pc$ be a set of qualitative characters. Then $\pc$ is compatible if and only if $\pig{\pc}$ has a proper minimal triangulation.
\end{theorem}

\begin{figure}[t]
\centering
\subfigure[$\xt{}$]{
   \scalebox{.7}{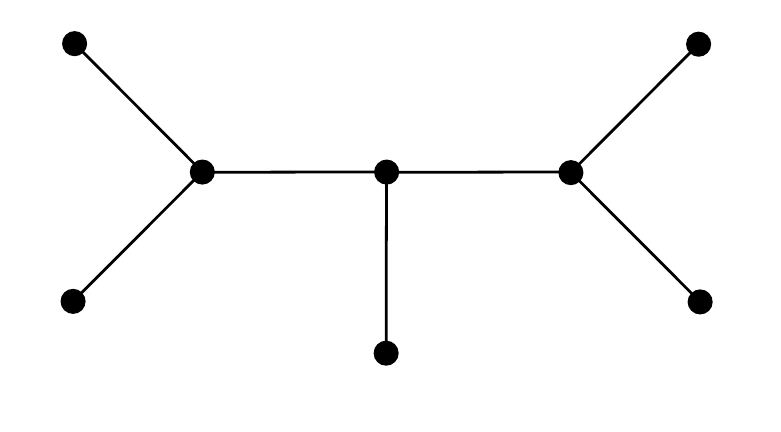}
 }
\subfigure[$\pig{\pc}$]{
   \scalebox{.7}{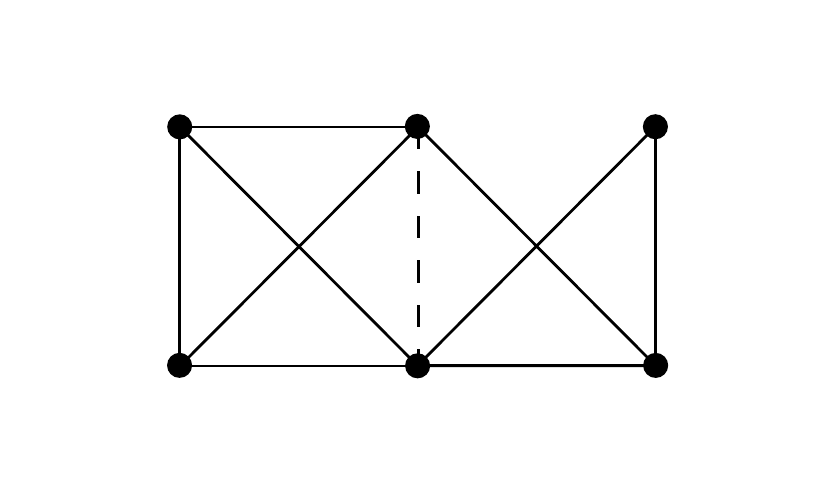}
 }
\caption{
An $X-$tree for $X = \{a,b,c,d,e\}$ and a non-chordal partition intersection graph $\pig{\pc}$.
The characters $\pc = \{\chi_1,\chi_2,\chi_3\}$ are $\chi_1 = ab|cd$, $\chi_2 = ac|bde$, and $\chi_3 = ab|de$.
Edges of $\pig{\pc}$ are given by solid edges, and the dashed edge represents the fill edge required to obtain the triangulation $\pig{\pc,\xt{}}$ of $\pig{\pc}$.
Both $\chi_1$ and $\chi_3$ are displayed by $\xt{}$, but $\chi_2$ is not because $\xt{}(ac)$ and $\xt{}(bde)$ intersect at $u$ and $v$.
This intersection induces the dashed fill edge of $\pig{\pc}$, which breaks $\chi_2$.
The edge $uv$ is distinguished by $\chi_1$.
}\label{fig_pp_pig}
\end{figure}

Two $X-$trees $\xt{}$ and $\xt{'}$ are \emph{isomorphic}, writing $\xt{} \isomorphic \xt{'}$, if there is a bijective map $\psi: V(T) \to V(T')$ that has the following properties:
\begin{enumerate}
	\item it preserves labels, meaning that $\phi' = \psi \circ \phi$; and
	\item it is a graph isomorphism, that is, $uv \in E(T)$ if and only if $\psi(u)\psi(v) \in E(T')$.
\end{enumerate}
A set of characters $\pc$ \emph{defines} a perfect phylogeny if it is the unique perfect phylogeny, up to isomorphism, for $\pc$.	
The \emph{unique perfect phylogeny problem} is to determine if a set $\pc$ of partial characters defines a perfect phylogeny.
If $\xtree{}$ displays $\chi$ and $uv$ is an edge of $T$ such that $u$ is a node of $\xt{}(A)$ and $v$ is a node of $\xt{}(A')$ where $A$ and $A'$ are distinct cells of $\chi$, then $uv$ is \emph{distinguished} by $\chi$.
If every edge of $\xt{}$ is distinguished by at least one character of $\pc$, then $\xt{}$ is \emph{distinguished} by $\pc$.
The following characterization is due to Semple and Steel.

\begin{theorem}\label{mainthm_semple_steel}\cite{SS02}
Let $\pc$ be a set of partial characters on $X$. 
Then $\pc$ defines a perfect phylogeny if and only if the following conditions are satisfied:
	\begin{description}
	\item[(a)] $\pig{\pc}$ has a unique proper minimal triangulation $H$; and
	\item[(b)] there is a free ternary perfect phylogeny for $\pc$ and it is distinguished by $\pc$.
	\end{description}
Further, if $\xt{}$ is the unique perfect phylogeny for $\pc$, then $\xt{}$ is a free ternary $X-$tree distinguished by $\pc$, and $\pig{\pc,\xt{}} = H$.
\end{theorem}

This result is the impetus of our current work, and one of our main interests is to re-formulate condition \textbf{(b)} in terms of combinatorial structures that play a significant role in the study of chordal graphs and minimal triangulations.

Chordal graphs are characterized by the existence of trees that represent the adjacency structure of the graph.
Suppose $G$ is a graph with vertices $V(G) = \{x_1, x_2, \ldots, x_n\}$.
A \emph{tree representation} of $G$ consists of a tree $T$ and subtrees $T_1, T_2, \ldots, T_n$ of $T$ such that two trees $T_i$ and $T_j$ intersect if and only if $x_i$ and $x_j$ are adjacent.
Here, the subtrees are in one-to-one correspondence with the vertex set of $G$, and this correspondence is made explicit by mapping each subtree $T_i$ to the vertex $x_i$ of $G$.
Observe that a node $v$ of $T$ defines a clique $\cliquemap(v) = \{x_i \mid v \mbox{ is a node of } T_i\}$ of $G$.
Notationally, we will write a tree representation as an ordered pair $\treerep{}$ where $\cliquemap{}$ maps nodes of $T$ to cliques of $G$ satisfying the following properties:
\begin{description}
	\item[(Edge Coverage)] a pair of vertices $x$ and $y$ of $G$ are adjacent if and only if there is a node $v$ of $T$ such that $x,y \in \cliquemap{}(v)$; and
	\item[(Convexity)] for each vertex $x$ of $G$, the set of nodes $\{v \in V(T) \mid x \in \cliquemap{}(v)\}$ induces a subtree of $T$ (i.e.\ a connected subgraph of $T$).
\end{description}

We will frequently refer to the convexity property throughout the paper.
As with $X-$trees, we will call $T$ the \emph{underlying tree} of $\tr{}$.
Often we will define a tree representation by only specifying the underlying tree $T$ and a collection of subtrees of $T$, which together implicitly define $\cliquemap{}$.

Let $T_1, T_2, \ldots, T_k$ be a collection of subtrees of $T$.
If each pair $T_i, T_j$ of subtrees intersect at a node $v_{ij}$, then by the \emph{Helly property} for subtrees of a tree \cite{G_Book04}, all of $T_1, T_2, \ldots, T_k$ intersect at a common node $v$.
This property manifests itself as a statement about cliques of $G$ and nodes of $\tr{}$ in the following way: for any clique $K$ of $G$, there is at least one node $u$ of $T$ such that $K \subseteq \cliquemap{}(u)$.
In particular, this is true when $K$ is a maximal clique of $G$ (i.e.\ no proper superset is also a clique), and therefore $\cliquemap{}(V(T))$ contains the set of maximal cliques of $G$.
If the maximal cliques of $G$ are in one-to-one correspondence with the nodes of $T$ via $\cliquemap{}$, then $\tr{}$ is a \emph{clique tree} of $G$.
See Figure \ref{fig_cg_ct} for an example.

\begin{theorem}\cite{B74,G74,W78}\label{thm_chordal_treereps}
	The following statements are equivalent.
\begin{description}
	\item[(a)] $G$ is a chordal graph.
	\item[(b)] $G$ has a tree representation.
	\item[(c)] $G$ has a clique tree.
\end{description}
\end{theorem}

Observe that if $\treerep{}$ is a clique tree of $G$ and $uv$ is an edge of $T$, then because $\cliquemap{}(u)$ and $\cliquemap{}(v)$ are both maximal cliques of $G$, there is a vertex $x$ of $G$ in $\cliquemap{}(u) - \cliquemap{}(v)$.
In general, a chordal graph has an exponential number of clique trees.
An algorithm to enumerate clique trees, along with a formula to count them, appears in \cite{HL89}.

We will often be analyzing a tree representation $\treerep{}$ of a triangulation $H$ of $\pig{\pc}$.
Given a vertex $(A,\chi)$ of $\pig{\pc}$, we will denote the subtree of $T$ that it corresponds to by $\tr{}(A,\chi)$.
Observe that $v$ is a node of $\tr{}(A,\chi)$ if and only if $(A,\chi) \in \cliquemap{}(v)$.
Given a set of characters $\pc$ on $X$ and an $X-$tree $\xtree{}$, a chordal graph $\pig{\pc,\xt{}}$ is given by adding an edge between two vertices $(A_1,\chi_1)$ and $(A_2,\chi_2)$ if and only if $\xt{}(A_1)$ and $\xt{}(A_2)$ intersect.
This construction, along with the fact that $\pig{\pc,\xt{}}$ is a triangulation of $\pig{\pc}$, is well-known in the phylogenetics literature, and will be discussed in detail in Section 2.

A chordal graph $G$ is \emph{uniquely representable} if it has a single clique tree, or \emph{ur-chordal} for short.
A ur-chordal graph is \emph{ternary} if each internal node of its clique tree has degree three, and its \emph{leafage}\footnote{In general, the leafage of a chordal graph is the minimum number of leaves that a clique tree of the graph can have \cite{LMW98}.} is the number of leaves its clique tree has.
Let $H$ be a proper triangulation of $\pig{\pc}$ and $\treerep{}$ be a clique tree of $H$.
An edge $uv$ of $\tr{}$ is \emph{incontractable with respect to $\chi$} if there are distinct cells $A$ and $A'$ of $\chi$ such that $u \in \tr{}(A,\chi)$ and $v \in \tr{}(A',\chi)$.
We say that $\tr{}$ is \emph{incontractable} with respect to $\pc$ if each edge is incontractable with respect to at least one $\chi \in \pc$.
Now we present our first main result.

\begin{theorem}\label{mainthm_urchordal_defining}
Suppose $\pc$ is a set of partial characters on $X$.
Then $\pc$ defines a perfect phylogeny if and only if the following conditions hold:
	\begin{description}
	\item[(a)] $\pig{\pc}$ has a unique proper minimal triangulation $H$;
	\item[(b)] $H$ is a ternary ur-chordal graph with leafage $|X|$; and
	\item[(c)] each edge of $H$'s unique clique tree is incontractable with respect to $\pc$.
	\end{description}
	Further, if $\xt{}$ is the perfect phylogeny defined by $\pc$, then $\xt{}$ is a free ternary $X-$tree distinguished by $\pc$, and $\pig{\pc,\xt{}} = H$.
\end{theorem}

Let $\pc$ be a set of partial characters on $X$ and $\chi \in \pc$.
Suppose $H$ is a triangulation of $\pig{\pc}$ with fill edge $(A,\chi)(A',\chi)$ where $A$ and $A'$ are distinct cells of $\chi$.
Then we say that $(A,\chi)(A',\chi)$ \emph{breaks} $\chi$, and $\chi$ is a \emph{broken character} of $H$.
For a triangulation $H$ of $\pig{\pc}$, its \emph{displayed characters} are the characters of $\pc$ that are not broken characters of $H$.
Bordewich, Huber, and Semple \cite{BHS05} proved that it is possible to find a maximum-sized compatible subset of $\pc$ using the partition intersection graph.

\begin{theorem}\label{thm_cr}
	Let $\pc$ be a set of partial characters on $X$. Then $\pc'$ is a maximum-sized compatible subset of $\pc$ if and only if there is a triangulation $H$ of $\pig{\pc}$ that has $\pc'$ as its displayed characters, and any other triangulation of $\pig{\pc}$ has at most $|\pc'|$ displayed characters.
\end{theorem}

A subset of partial characters $\pc' \subseteq \pc$ is a \emph{maximal defining subset} of $\pc$ when $\pc'$ defines a perfect phylogeny, and there is no compatible set $\pc''$ such that $\pc' \subset \pc'' \subset \pc$.
Our second main result is the following.

\begin{theorem}\label{mainthm_subset_urchordal_defining}
	Suppose $\pc$ is a set of partial characters on $X$ and $\pc' \subseteq \pc$. Then $\pc'$ is a maximal defining subset of $\pc$ if and only if the following conditions hold:
	\begin{description}
	\item[(a)] $\pig{\pc}$ has a unique minimal triangulation $H$ that has $\pc'$ as its displayed characters, and no other minimal triangulation of $\pig{\pc}$ has at least $\pc'$ as its displayed character set;
	\item[(b)] $H$ is a ternary ur-chordal graph with leafage $|X|$; and
	\item[(c)] each edge of $H$'s unique clique tree is incontractable with respect to $\pc'$.
	\end{description}
	Further, if $\xt{}$ is the perfect phylogeny defined by $\pc'$, then $\xt{}$ is a free ternary $X-$tree distinguished by $\pc'$, and $\pig{\pc,\xt{}} = H$.
\end{theorem}

\section{Chordal Graph Preliminaries}

\begin{figure}[t]
\centering
\subfigure[$G$]{
   \scalebox{.78}{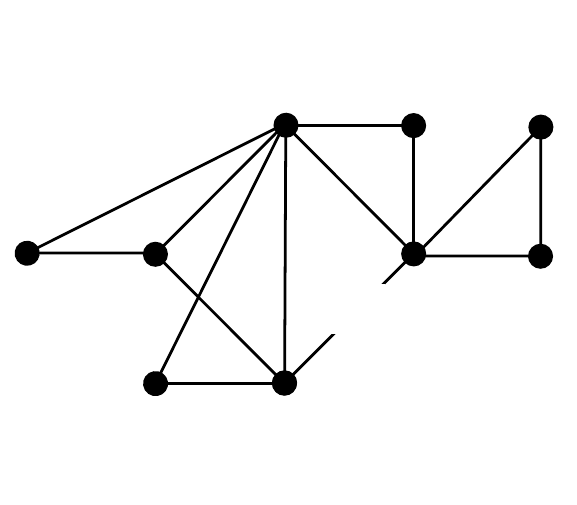}
 }
\subfigure[$\tr{}$]{
   \scalebox{.78}{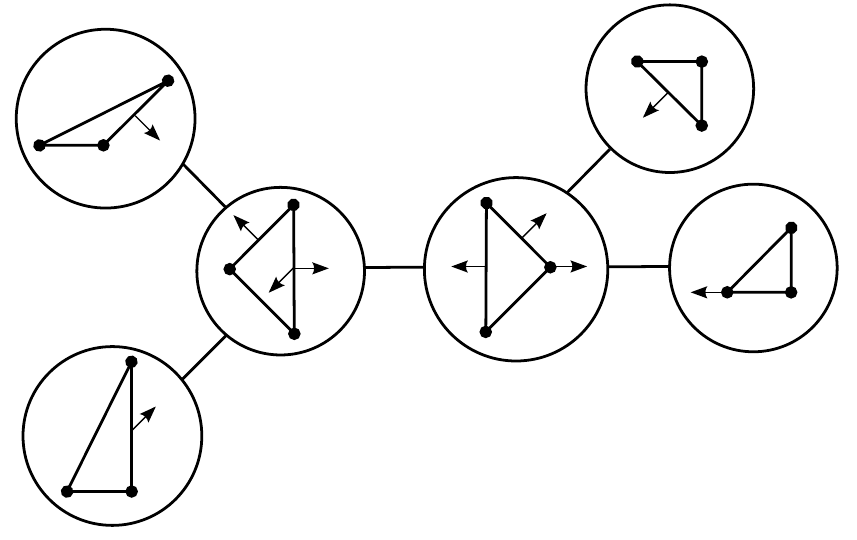}
 }
\caption{
A chordal graph $G$ and clique tree $\treerep{}$ of $G$.
There are three clique trees for $G$, one of which is obtained by removing the bottom-leftmost node and its incident edge, and then attaching it to $v$.
The maximal clique map $\cliquemap{}$ is defined by the triangle drawn inside each node.
Arrows indicate vertices that consist of the intersection of a neighboring node's maximal clique, and each such intersection is a minimal separator by Theorem \ref{thm_minseps_ctedges}.
}\label{fig_cg_ct}
\end{figure}

In this section, we detail known results on chordal graphs that are necessary for the remainder of the paper. 
Suppose $G = (V,E)$ is a graph and $S \subseteq V$. 
Let $G - S$ denote the graph obtained from $G$ by removing $S$ and all edges incident to at least one vertex in $S$. 
If there are vertices $x$ and $y$ of $G$ that are connected in $G$ but not in $G - S$, then $S$ is an \emph{$xy-$separator}, and if no proper subset of $S$ has this property, it is a \emph{minimal $xy-$separator}.
When $S$ is a minimal $xy-$separator for at least one pair of vertices $x$ and $y$, it is a \emph{minimal separator}.
Minimality in this definition is relative; it is possible to have containment relationships between two minimal separators\footnote{Dirac \cite{D61} called these sets \emph{relatively minimal cut-sets}, which is perhaps more descriptive, but this term has stuck of the modern literature on chordal graphs and minimal triangulations.}.
The maximal connected subsets of $G - S$ are the \emph{connected components} of $G - S$.
Let $C$ be a connected component of $G - S$.
The \emph{neighborhoood} of $C$ in $G$, denoted $N(C)$, is the set of vertices of $S$ that are adjacent to at least one vertex in $C$.
If $N(C) = S$, then it is a \emph{full component} of $G - S$.
The following useful characterization of minimal separators is well-known, and left as an exercise in \cite{G_Book04}.

\begin{lemma}\label{lem_minsep_iff_two_full_components}
	Let $G = (V,E)$ be a graph, and $S \subseteq V$. Then $S$ is a minimal separator if and only if $G - S$ has two or more full components.
\end{lemma}

A minimal separator has \emph{multiplicity} $k$ if $G - S$ has $k-1$ full components.
Interestingly, clique trees contain detailed information about the minimal separators of the graph it represents, which will be useful for our proofs in later sections.

\begin{theorem}\cite{HL89,MM_Book99}\label{thm_minseps_ctedges}
Suppose $G = (V,E)$ is a chordal graph, $\treerep{}$ is a clique tree of $G$, and $S \subseteq V$.
Then $S$ is a minimal separator of $G$ if and only if there is an edge $uv$ of $T$ such that $S = \cliquemap{}(u) \cap \cliquemap{}(v)$\footnote{This theorem also implies that a minimal separator of a chordal graph is a clique.
In fact, chordal graphs are characterized by having only clique minimal separators, which is one of the earliest results on chordal graphs \cite{D61}.}.
Further, the multiplicity of $S$ is the number of edges of $T$ with this property.
\end{theorem}

We will not need all of the following characterizations of ur-chordal graphs, but we list them for completeness.

\begin{theorem}\cite{KM02,HT06}\label{thm_urchordal_characterizations}
Let $G$ be a chordal graph.
Then the following statements are equivalent.
\begin{description}
	\item[a)] $G$ is uniquely representable.
	\item[b)] If $S$ is a minimal separator of $G$, then there are exactly two maximal cliques $K$ and $K'$ of $G$ such that $S \subseteq K,K'$.
	\item[c)] Each minimal separator of $G$ has multiplicity one.
	\item[d)] There is no minimal separator of $G$ that properly contains another minimal separator of $G$.
	\item[e)] The number of minimal separators of $G$ is the number of maximal cliques of $G$ minus one.
\end{description}
\end{theorem}

Chordal graphs can be recognized in linear time \cite{BP92}, and the maximal cliques and minimal separators of a chordal graph may be computed in linear time \cite{BP11}. 
Using property \textbf{e)} of Theorem \ref{thm_urchordal_characterizations}, this allows the recognition of ur-chordal graphs in linear time.

\section{$X-$trees, Clique Trees, and Tree Representations}

In this section, we define two operations to facilitate the discussion between $X-$trees and tree representations, and prove or state results that will be useful in later sections.
Both operations are commonly used in the literature (see \cite{S92,SS02}), but do not seem to be named as we will do here.
The results in this section will be useful for proving our characterization for maximal defining subsets.

Given a set of characters $\pc$ on $X$ and an $X-$tree $\xtree{}$, construct the chordal graph $\pig{\pc,\xt{}}$ having:
\begin{description}
	\item[$\bullet$] vertex set identical to that of $\pig{\pc}$; and
	\item[$\bullet$] an edge between $(A,\chi)$ and $(A',\chi')$ if and only if $\xt{}(A)$ and $\xt{}(A')$ intersect.
\end{description}
This graph has a tree representation $\tr{}$ with underlying tree $T$ and subtrees obtained by defining $\tr{}(A,\chi) = \xt{}(A)$ for each vertex $(A,\chi)$ of $\pig{\pc}$.
Therefore $\pig{\pc}$ is chordal by Theorem \ref{thm_chordal_treereps}.
Each edge $(A_1,\chi_1)(A_2,\chi_2)$ of $\pig{\pc}$ has cells $A_1$ and $A_2$ that share at least one member of $X$, say $a$, so $\xt{}(A_1)$ and $\xt{}(A_2)$ intersect at $\phi(a)$.
Therefore $\tr{}(A_1,\chi_1)$ and $\tr{}(A_2,\chi_2)$ intersect at $\phi(a)$ as well, so $(A_1,\chi_1)(A_2,\chi_2)$ is an edge of $\pig{\pc,\xt{}}$ by subtree intersection.
Hence each edge of $\pig{\pc}$ is also an edge of $\pig{\pc,\xt{}}$, so $\pig{\pc,\xt{}}$ is a triangulation of $\pig{\pc}$.
We will say that $\xt{}$ \emph{derives} the tree representation $\treerep{}$ of $\pig{\pc,\xt{}}$ and $\tr{}$ is \emph{derived from} $\xt{}$.
In general, $\tr{}$ is not a clique tree.

\begin{figure}[t]
\centering
\scalebox{1}{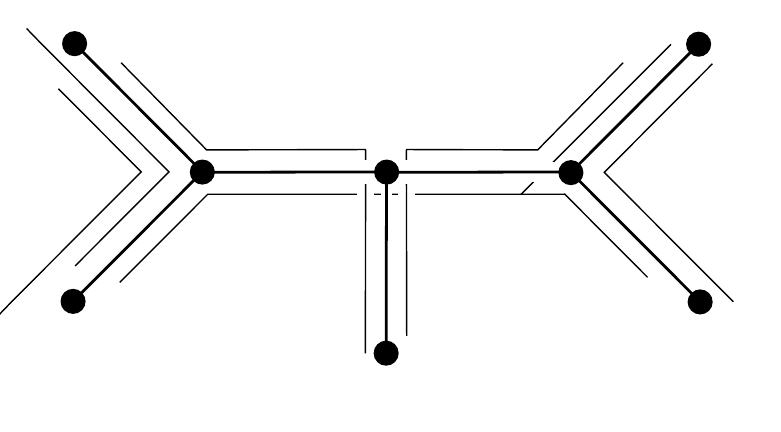}
\caption{
The tree representation $\tr{}$ of $\pig{\pc,\xt{}}$ derived from $\xt{}$ in Figure \ref{fig_pp_pig}.a.
The triangulation $\pig{\pc,\xt{}}$ of $\pig{\pc}$ is depicted in Figure \ref{fig_pp_pig}.b; note that any two subtrees of $\tr{}$ intersect precisely when the corresponding vertices of $\pig{\pc,\xt{}}$ are adjacent.
Observe that $\tr{}$ is not a clique tree; for example, the two left-most nodes map to non-maximal cliques.
Additionally, there are two nodes that map to the maximal clique $\{ (cd,\chi_1), (bde,\chi_2), (de,\chi_3) \}$.
Obtaining a clique tree from a tree representation is described in \cite{G74}.
}\label{fig_derived}
\end{figure}

\begin{observation}\label{obs_derived_properties}
Let $\pc$ be a set of partial characters on $X$, $\xt{}$ be an $X-$tree, and $\tr{}$ be the tree representation induced by $\xt{}$.
Then the underlying tree of $\tr{}$ is the underlying tree of $\xt{}$, and for all $\chi \in \pc$ and cells $A$ of $\chi$, $\xt{}(A) = \tr{}(A,\chi)$.
Further, $\pig{\pc,\xt{}}$ is a triangulation of $\pig{\pc}$.
\end{observation}

\begin{lemma}\label{lem_defxtree_is_cliquetree}
Let $\pc$ be a set of partial characters on $X$, and $\xt{}$ be an $X-$tree that displays $\pc' \subseteq \pc$.
Suppose that each edge of $\xt{}$ is distinguished by $\pc'$.
Then the tree representation of $\pig{\pc,\xt{}}$ derived from $\xt{}$ is a clique tree of $\pig{\pc,\xt{}}$.
\end{lemma}
\begin{proof}
Let $\treerep{}$ be the tree representation of $\pig{\pc,\xt{}}$ derived from $\xt{}$.
For the sake of contradiction, assume that $\tr{}$ is not a clique tree, so that $\cliquemap{}$ is not a one-to-one map between the nodes of $T$ and the maximal cliques of $\pig{\pc,\xt{}}$.
Then there must be nodes $u$ and $v$ of $T$ such that $\cliquemap{}(v) \subseteq \cliquemap{}(u)$.
Let $v'$ be the closest node to $v$ between $u$ and $v$ in $T$, allowing the possibility that $v' = u$.
Note that each vertex in $\cliquemap{}(v)$ is also a vertex of $\cliquemap{}(v')$ by convexity, so $\cliquemap{}(v) \subseteq \cliquemap{}(v')$.

Now, $vv'$ is distinguished by some $\chi \in \pc'$, so there are distinct cells $A$ and $A'$ of $\chi$ such that $v$ is a node of $\xt{}(A)$ and $v'$ is a node of $\xt{}(A')$.
By Observation \ref{obs_derived_properties}, $v$ is a node of $\tr{}(A,\chi)$, so $(A,\chi) \in \cliquemap{}(v)$ and by containment $(A,\chi) \in \cliquemap{}(v')$.
Further, $(A',\chi) \in \cliquemap{}(v')$, so $(A,\chi)(A',\chi)$ is a fill edge of $\pig{\pc,\xt{}}$ and it breaks $\chi$.
This contradicts the assumption that $\xt{}$ displays $\pc'$, so $\tr{}$ must be a clique tree of $\pig{\pc,\xt{}}$.
\qed\end{proof}

\begin{lemma}\label{lem_define_implies_urchordal}
Let $\pc$ be a set of partial characters on $X$ and $\xt{}$ be an $X-$tree that displays $\pc' \subseteq \pc$.
Suppose that $\xt{}$ is free, ternary, and each edge of $\xt{}$ is distinguished by $\pc'$.
Then $\pig{\pc,\xt{}}$ is uniquely representable.
\end{lemma}
\begin{proof}
To prove that $\pig{\pc,\xt{}}$ is uniquely representable, we use Theorem \ref{thm_urchordal_characterizations} and show that there are no containment relationships between the minimal separators of $\pig{\pc,\xt{}}$.
Working towards a contradiction, assume that $S \subset S'$ are minimal separators of $\pig{\pc,\xt{}}$, and let $\treerep{}$ be the tree representation derived from $\xt{}$.
Then $\tr{}$ is a clique tree by Lemma \ref{lem_defxtree_is_cliquetree}, and there are edges $uv$ and $u'v'$ of $T$ such that $S = \cliquemap{}(u) \cap \cliquemap{}(v)$ and $S' = \cliquemap{}(u') \cap \cliquemap{}(v')$ by Theorem \ref{thm_minseps_ctedges}.
Without loss of generality, assume that the path from $v$ to $v'$ does not contain either $u$ or $u'$ (perhaps $v = v'$).
Let $w$ be the node on this path adjacent to $v$ if $v \neq v'$, otherwise let $w = u'$.
In either case, $\cliquemap{}(w)$ contains $S$: if $w = u'$, then $S \subseteq \cliquemap{}(w)$.
Otherwise, each $(A,\chi)$ in $S$ is an element of both $\cliquemap{}(u)$ and $\cliquemap{}(u')$, so $(A,\chi) \in \cliquemap{}(w)$ by convexity, and $S \subseteq \cliquemap{}(w)$.
	
To complete the proof, we will obtain a contradiction by showing that $S$ has a vertex not in $\cliquemap{}(w)$.
There is a character $\chi'$ in $\pc'$ that distinguishes $wv$, and distinct cells $A'$ and $A''$ of $\chi'$ such that $v$ is a node of $\xt{}(A')$ and $w$ is a node of $\xt{}(A'')$.
Now, $v$ has at least two neighbors, and because $\xt{}$ is ternary, $v$ must have degree three.
Also, $v$ is not mapped to by $\phi$ because $\xt{}$ is free, so in order for $v$ to be a node of $\xt{}(A')$, there must be at least two nodes of $\xt{}$ in $\phi(A')$ that are not $v$, and the path between these two nodes must contain $v$.
This path must also contain two of $v$'s neighbors, and neither of these vertices can be $w$, because $w$ is not a node of $\xt{}(A')$.
Thus $u$ must be on this path, so it is a node of $\xt{}(A')$.
Both $\cliquemap{}(u)$ and $\cliquemap{}(v)$ contain $(A',\chi')$ by Observation \ref{obs_derived_properties}, so it must be a vertex of $S = \cliquemap{}(u) \cap \cliquemap{}(v)$.
Further, $(A',\chi') \notin \cliquemap{}(w)$ because $w$ is not a node of $\xt{}(A') = \tr{}(A',\chi')$.
This is impossible because we have shown that both $(A',\chi') \in S - \cliquemap{}(w)$ and $S \subseteq \cliquemap{}(w)$.
Thus there are no containment relationships between the minimal separators of $\pig{\pc,\xt{}}$, and $\pig{\pc,\xt{}}$ is uniquely representable by Theorem \ref{thm_urchordal_characterizations}.
\qed\end{proof}

\begin{lemma}\cite{BHS05}\label{lem_xtree_displays_legalcharacters}
Let $\pc$ be a set of partial characters on $X$, $\xt{}$ an $X-$tree, and $\pc'$ be the subset of $\pc$ displayed by $\xt{}$.
Then $\pig{\pc,\xt{}}$ is a triangulation of $\pig{\pc}$ in which the displayed characters are $\pc'$.
\end{lemma}

The previous three lemmas can be summarized as follows.

\begin{theorem}\label{thm_almostdefine_urchordal}
Let $\pc$ be a set of partial characters on $X$ and $\xt{}$ be an $X-$tree that displays $\pc' \subseteq \pc$.
Suppose that $\xt{}$ is free, ternary, and each edge of $\xt{}$ is distinguished by $\pc'$.
Then $\pig{\pc,\xt{}}$ is a uniquely representable chordal graph, and the tree representation derived from $\xt{}$ is its unique clique tree.
Further, the displayed characters of $\pig{\pc,\xt{}}$ are exactly $\pc'$.
\end{theorem}

Now suppose that $\treerep{}$ is a clique tree of a triangulation $H$ of $\pig{\pc}$, with the goal of defining an $X-$tree $\xt{}$.
The discussion we provide here is standard, e.g.\ see \cite{SS_Book03,BHS05}.
Construct a map $\phi$ from $X$ to $T$ by defining, for each $a \in X$, $\phi(a) = v$ if and only if $\cliquemap(v)$ contains every vertex of $\pig{\pc}$ whose cell contains $a$.
These vertices form a clique because $a$ is contained in each of their cells.
Because we have only added fill edges to obtain $H$, this clique a subset of a maximal clique of $H$, and hence $v$ exists.
There may be more than one choice for $v$, each of which we call a \emph{candidate node} for $a$.
Let $u$ be a leaf of $T$ with neighbor $w$.
Then $\cliquemap{}(u)$ is a maximal clique that contains a vertex $(A,\chi)$ of $\pig{\pc}$ that is not found in $\cliquemap{}(w)$.
By convexity, $u$ is the only node of $T$ whose corresponding maximal clique contains $(A,\chi)$.
Each $a' \in A$ has $u$ as its unique candidate node, and hence every leaf of $T$ is a unique candidate node for at least one element of $X$.
Thus each leaf of $T$ must be labeled by $\phi$.
To finish constructing an $X-$tree, obtain $T'$ by suppressing any unlabeled nodes of $T$ that have degree two.
The result is an $X-$tree $\xt{} = (T',\phi)$, and we say that $\tr{}$ \emph{induces} $\xt{}$ and $\xt{}$ is \emph{induced by} $\tr{}$.
We emphasize that the underlying tree of $\xt{}$ need not be the same as the underlying tree of $\tr{}$.
Note that, because an element of $X$ may have multiple candidate nodes, $\tr{}$ may induce multiple $X-$trees.
Next, we show that when $H$ is a minimal triangulation, much of $\xt{}$'s structure is described by $\tr{}$.
The following lemma will be useful.

\begin{lemma}\label{lem_mintri_fill_ms}\cite{PS97,KKS97}
Let $G$ be a graph and $H$ be a minimal triangulation of $G$.
If $uv$ is a fill edge of $H$, then there is a minimal separator of $H$ that contains both $u$ and $v$.
\end{lemma}

\begin{figure}[t]
\centering
\subfigure[$\tr{'}$]{
   \scalebox{.78}{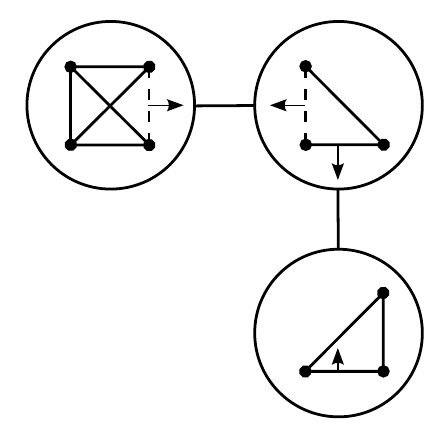}
 }
\subfigure[$\xt{'}$]{
   \scalebox{.78}{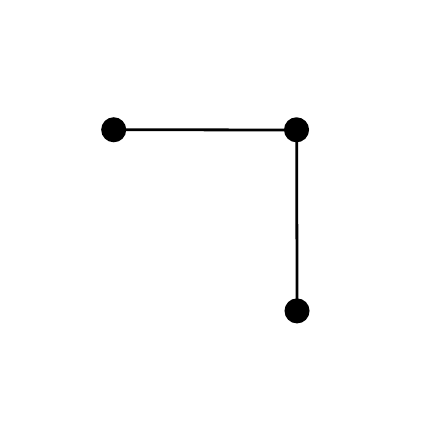}
 }
\caption{
A clique tree $\tr{'}$ of $\pig{\pc,\xt{}}$ from Figure \ref{fig_pp_pig}.b and an $X-$tree $\xt{'}$ induced by $\tr{'}$.
Note that $\pig{\pc,\xt{'}} = \pig{\pc,\xt{}}$, where $\xt{}$ is the $X-$tree from Fig.\ \ref{fig_pp_pig}.a.
}\label{fig_induced}
\end{figure}

Though not stated in this form, the following lemma follows from the proof of Lemma 2.4 and the statement of Corollary 2.5 in \cite{SS02}.

\begin{lemma}\label{lem_mintri_equals_xtint}
Let $H$ be a minimal triangulation of $\pig{\pc}$, and suppose $\xt{}$ is induced by a clique tree of $H$.
Then $H = \pig{\pc,\xt{}}$.
\end{lemma}

\begin{lemma}\label{lem_mintri_underlying_tree}
Let $H$ be a minimal triangulation of $\pig{\pc}$, $\treerep{}$ be a clique tree of $H$, and suppose $\tr{}$ induces $\xt{}$.
Then the underlying tree of $\xt{}$ is $T$.
\end{lemma}
\begin{proof}
We have already seen that every leaf of $T$ is the unique candidate node of some element of $X$.
In addition to this, it was also shown in \cite{GLG12} that if $u$ is a node of $\tr{}$ of degree two, then $u$ is the unique candidate node of some element of $X$.
This was done by showing that $\cliquemap(u)$ contains either:
\begin{description}
	\item[1.]	a vertex $(A_1,\chi_1)$ of $\pig{\pc}$ that is not contained in $\cliquemap(w)$ for any other node $w \neq u$ of $T$; or
	\item[2.]	an edge $(A_2,\chi_2) (A_3,\chi_3)$ of $\pig{\pc}$, whose incident vertices have cells with non-empty intersection, and are not both contained in $\cliquemap(w)$ for any other node $w \neq u$ of $T$.
\end{description}

For completeness, we outline a proof here.
Using convexity and the fact that $u$ has degree two, it follows that either a unique vertex or unique pair of vertices are contained in $\cliquemap(u)$.
It remains to show that $(A_2,\chi_2) (A_3,\chi_3)$ is actually an edge of $\pig{\pc}$ (so $A_2 \cap A_3$ is non-empty).
If not, then by Lemma \ref{lem_mintri_fill_ms} there is a minimal separator $S$ of $H$ containing both $(A_2,\chi_2)$ and $(A_3,\chi_3)$.
By Theorem \ref{thm_minseps_ctedges}, there is an edge $u_1u_2$ of $T$ such that $S = \cliquemap{}(u_1) \cap \cliquemap{}(u_2)$.
But this contradicts case 2, so it must be that $(A_2,\chi_2) (A_3,\chi_3)$ is an edge of $\pig{\pc}$.

In both cases, $u$ is the unique candidate node of some element in $X$ (this element is either $a \in A_1$ or $a \in A_2 \cap A_3$), so every degree two node of $T$ is labeled by $\phi$, and there are no nodes of $T$ that need to be suppressed.
Hence the underlying tree of $\xt{}$ is $T$.
\qed\end{proof}

\begin{lemma}\label{lem_mintri_subtrees}
Let $H$ be a minimal triangulation of $\pig{\pc}$, $\treerep{}$ be a clique tree of $H$, and suppose $\tr{}$ induces $\xt{}$.
Then for each vertex $(A,\chi)$ of $\pig{\pc}$, $\xt{}(A) = \tr{}(A,\chi)$.
\end{lemma}
\begin{proof}
Let $(A,\chi)$ be a vertex of $\pig{\pc}$ and consider a node $v$ of $\xt{}(A)$.
Either $v = \phi(a)$ for some $a \in A$ or $v$ lies between $\phi(a_1)$ and $\phi(a_2)$ for some $a_1, a_2 \in A$.
In the first case, $(A,\chi) \in \cliquemap{}(v)$ because $v$ is a candidate node for $a$.
Similarly, in the second case, $(A,\chi)$ is an element of both $\cliquemap{}(\phi(a_1))$ and $\cliquemap{}(\phi(a_2))$, and therefore $(A,\chi) \in \cliquemap{}(v)$ by convexity.
In both cases $v$ is a node of $\tr{}(A,\chi)$, so $\xt{}(A) \subseteq \tr{}(A,\chi)$.

To finish proving equality, suppose that $\xt{}(A) \subset \tr{}(A,\chi)$.
Define a tree representation $\treerep{'}$ of a graph $H'$ as follows: set $T' = T$, and define subtrees $\tr{'}(A',\chi')$ of $T$ for each vertex $(A',\chi')$ of $\pig{\pc}$ as follows:
\begin{description}
	\item[1.] $\tr{'}(A',\chi') = \tr{}(A',\chi')$ if $(A',\chi') \neq (A,\chi)$, and
	\item[2.] $\tr{'}(A',\chi') = \xt{}(A)$ if $(A',\chi') = (A,\chi)$.
\end{description}
We have already seen that $\tr{'}(A',\chi') \subseteq \tr{}(A',\chi')$ for every vertex $(A',\chi')$ of $\pig{\pc}$, so the edge set of $H'$ is a subset of the edge set of $H$.
If $(A_1,\chi_1)(A_2,\chi_2)$ is an edge of $\pig{\pc}$, then $A_1$ and $A_2$ have at least one element $a$ in common, and thus $\xt{}(A_1)$ and $\xt{}(A_2)$ intersect at $\phi(a)$.
Further, $\xt{}(A') \subseteq \tr{'}(A',\chi')$ for each vertex $(A',\chi')$ of $\pig{\pc}$, so $\tr{'}(A_1,\chi_1)$ and $\tr{'}(A_2,\chi_2)$ also intersect at $\phi(a)$.
Therefore $(A_1,\chi_1)(A_2,\chi_2)$ is an edge of $H'$, and $H'$ is chordal by Theorem \ref{thm_chordal_treereps}, so it is a triangulation of $\pig{\pc}$.

To complete the proof, we show that $H'$ must have an edge that does not exist in $H$.
Because $\xt{}(A) \subset \tr{}(A,\chi)$, there is a node $u$ of $\tr{}(A,\chi) - \xt{}(A)$ that is adjacent to a node $w$ of $\xt{}(A) = \tr{'}(A,\chi)$.
By maximality, there is a vertex $(A',\chi') \in \cliquemap{}(u) - \cliquemap{}(w)$, and because $u$ is a node of $\tr{}(A,\chi)$, $(A,\chi) \in \cliquemap{}(u)$.
Therefore $(A',\chi')(A,\chi)$ is an edge of $H$.
The situation in $H'$ is different: if $(A',\chi')(A,\chi)$ is an edge of $H'$, then $\tr{'}(A,\chi)$ and $\tr{'}(A',\chi')$ intersect at a node $v'$.
But if $v'$ is a node of $\tr{'}(A,\chi)$, then $w$ is on the path from $v'$ to $u$, and by convexity this would imply that $(A',\chi') \in \cliquemap{}(w)$.
Hence $(A',\chi')(A,\chi)$ is not an edge of $H'$, so the edge set of $H'$ is a proper subset of the edge set of $H$.
This is impossible because $H$ is a minimal triangulation of $\pig{\pc}$, so it must be that $\xt{}(A) = \tr{}(A,\chi)$ for each vertex $(A,\chi)$ of $\pig{\pc}$.
\qed\end{proof}

Lemmas \ref{lem_mintri_equals_xtint}, \ref{lem_mintri_underlying_tree}, and \ref{lem_mintri_subtrees} are summarized below.

\begin{theorem}\label{thm_mintri_induced_xt}
Let $H$ be a minimal triangulation of $\pig{\pc}$, $\tr{}$ be a clique tree of $H$, and suppose $\tr{}$ induces $\xt{}$.
Then the underlying tree of $\xt{}$ is the underlying tree of $\tr{}$, and for each vertex $(A,\chi)$ of $\pig{\pc}$, $\xt{}(A) = \tr{}(A,\chi)$.
Further, $H = \pig{\pc,\xt{}}$.
\end{theorem}

\section{Maximal Defining Subsets of Characters}

This section is devoted to the proof of Theorem \ref{mainthm_subset_urchordal_defining}.
Its proof will follow mainly from Propositions \ref{mainprop_2} and \ref{mainprop_1}.
Recall that, for a graph $H$ and a subset $U$ of its vertices, the graph $H - U$ is obtained by removing the vertices in $U$ and edges of $H$ incident to one or more vertices of $U$.

\begin{lemma}\label{lem_subtri_to_supertri}
Let $\pc$ be a set of partial characters and $\pc' \subseteq \pc$.
Suppose $H'$ is a minimal triangulation of $\pig{\pc'}$, and let $U$ be the vertices of $\pig{\pc}$ that are not vertices of $\pig{\pc'}$.
Then there is a minimal triangulation $H$ of $\pig{\pc}$ such that $H' = H - U$.
\end{lemma}
\begin{proof}
	Let $H'$ be a minimal triangulation of $\pig{\pc'}$ and $H^*$ be the graph obtained by adding the following fill edges to $\pig{\pc}$:
\begin{description}
	\item[1.] the fill edges of $H'$; and
	\item[2.] fill edges of the form $(A,\chi)(A',\chi')$ where $(A,\chi)$ is a vertex of $U$, and $(A',\chi')$ is any vertex of $\pig{\pc}$.
\end{description}
First we prove that $H^*$ is chordal, and then we will use it to construct $H$, a minimal triangulation of $\pig{\pc}$ such that $H - U = H'$.
	
Let $(A_1,\chi_1), (A_2,\chi_2), \ldots, (A_k,\chi_k)$ be a cycle in $H^*$.
If $\chi_i \in \pc'$ for all $i = 1, 2, \ldots, k$, then this cycle is also a cycle of $H'$, and therefore has a chord that is an edge of $H'$.
Each edge of $H'$ is also an edge of $H^*$, so this cycle has a chord in $H^*$.
Otherwise, without loss of generality, $\chi_1 \in \pc - \pc'$ and $(A_1,\chi_1) \in U$ so either $(A_1,\chi_1)(A_3,\chi_3)$ is an edge of $\pig{\pc}$ or is a fill edge of $H^*$ of type \textbf{2}.
In either case, this cycle has a chord, so $H^*$ is chordal.

Now let $H$ be any minimal triangulation of $\pig{\pc}$ such that the edge set of $H$ is a subset of the edge set of $H^*$.
Every edge of $H - U$ is either an edge of $\pig{\pc'}$ or is an edge of $H'$ by the construction of $H^*$ and $H$.
Therefore the edge set of $H - U$ is a subset of the edge set of $H'$.
Further, $H - U$ is chordal because any cycle of $H - U$ is a cycle of $H$ (i.e.\ chordality is \emph{inherited} \cite{R70}), so it is a triangulation of $\pig{\pc'}$.
By minimality of $H'$, it must be that the edge set of $H - U$ is equal to the edge set of $H'$, so $H' = H - U$.
\qed\end{proof}

\begin{lemma}\cite{GG11} see also \cite{BHS05}\label{lem_legalcharacters_mintri}
Let $\pc$ be a set of partial characters on $X$ and suppose that $\pc'$ is a compatible subset of $\pc$.
Then there is a minimal triangulation of $\pig{\pc}$ whose displayed characters are at least $\pc'$.
\end{lemma}

Though not stated in this form, the following lemma is a direct result of Lemma 5.1 in \cite{BHS05} and its proof.

\begin{lemma}\label{lem_triangulation_cliquetree_displays_legalcharacters}
Let $\pc$ be a set of partial characters on $X$.
Suppose $H$ is a triangulation of $\pig{\pc}$ with displayed characters $\pc'$.
Then if $\xt{}$ is induced by a clique tree of $H$, it is a perfect phylogeny for $\pc'$.
\end{lemma}

\begin{proposition}\label{mainprop_2}
Let $\pc$ be a set of partial characters on $X$ and $\pc' \subseteq \pc$.
Suppose that the following conditions hold:
	\begin{description}
	\item[(i)] $\pig{\pc}$ has a unique minimal triangulation $H$ that has $\pc'$ as its displayed characters, and no other minimal triangulation of $\pig{\pc}$ has at least $\pc'$ as its displayed character set;
	\item[(ii)] $H$ is a ternary ur-chordal graph with leafage $|X|$; and
	\item[(iii)] each edge of $H$'s unique clique tree is incontractable with respect to $\pc'$.
	\end{description}
Then $\pc'$ is a maximal defining subset of $\pc$.
\end{proposition}
\begin{proof}
We begin by showing that $\pc'$ has a unique perfect phylogeny using Theorem \ref{mainthm_semple_steel}, and finish the proof showing that no superset of $\pc'$ has a unique perfect phylogeny.

Throughout the proof, $H$ will denote the unique minimal triangulation of $\pig{\pc}$ given by \textbf{(i)} whose displayed character set is $\pc'$.
By Lemma \ref{lem_triangulation_cliquetree_displays_legalcharacters}, $\pc'$ has a perfect phylogeny so it is compatible.
There is a proper triangulation of $\pig{\pc'}$ by Theorem \ref{thm_buneman}, so $\pig{\pc'}$ has a proper minimal triangulation as well.
To see that condition \textbf{(a)} of Theorem \ref{mainthm_semple_steel} holds, suppose that $H'_1$ and $H'_2$ are proper minimal triangulations of $\pig{\pc'}$ given by Theorem \ref{thm_buneman}, and let $U$ be the set of vertices of $\pig{\pc}$ not in $\pig{\pc'}$.
By Lemma \ref{lem_subtri_to_supertri}, there are minimal triangulations $H_1$ and $H_2$ of $\pig{\pc}$ that satisfy $H'_1 = H_1 - U$ and $H'_2 = H_2 - U$.
The displayed characters of $H_1$ and $H_2$ must be at least $\pc'$, because any fill edge that breaks a character of $\pc'$ would also appear in $H'_1$ or $H'_2$, and both $H'_1$ and $H'_2$ are proper triangulations of $\pig{\pc'}$.
By \textbf{(i)}, we have $H_1 = H = H_2$.
Therefore $H'_1 = H - U = H'_2$, so condition \textbf{(a)} of Theorem \ref{mainthm_semple_steel} is satisfied with respect to $\pc'$.

Now we show condition \textbf{(b)} of Theorem \ref{mainthm_semple_steel} holds.
Let $\tr{}$ be the unique clique tree of $H$ given by \textbf{(ii)}, and suppose $\xt{}$ is the $X-$tree induced by $\tr{}$.
$\xt{}$ displays $\pc'$ by Lemma \ref{lem_triangulation_cliquetree_displays_legalcharacters}, and by Theorem \ref{thm_mintri_induced_xt} the underlying tree $T$ of $\tr{}$ is also the underlying tree of $\xt{}$.
Further, $T$ is ternary and has $|X|$ leaves by \textbf{(ii)}, so $\xt{}$ must be free and ternary.
To see that $\xt{}$ is distinguished by $\pc'$, consider an edge $uv$ of $T$.
By \textbf{(iii)} $uv$ is incontractible with respect to $\pc'$, so there is a character $\chi \in \pc'$ and distinct cells $A$ and $A'$ of $\chi$ such that $u$ is a node of $\tr{}(A,\chi)$ and $v$ is a node of $\tr{}(A',\chi)$.
But $\xt{}(A) = \tr{}(A,\chi)$ and $\xt{}(A') = \tr{}(A',\chi)$ by Theorem \ref{thm_mintri_induced_xt}, so $\chi$ distinguishes $uv$.
Hence condition \textbf{(b)} of Theorem \ref{mainthm_semple_steel} also holds with respect to $\pc'$, so $\pc'$ defines $\xt{}$.

Last, we show that no proper superset of $\pc'$ also defines an $X-$tree.
If any superset $\pc^*$ of $\pc'$ was compatible, then by Lemma \ref{lem_legalcharacters_mintri} some minimal triangulation of $\pig{\pc}$ has at least $\pc^*$ as its displayed character set.
This would contradict \textbf{(i)}, so no such superset can exist.
This completes the proof.
\qed\end{proof}

\begin{lemma}\label{lem_maxdef_uniqueness}
Suppose $\pc'$ is a maximal defining subset of $\pc$, and $H$, $H'$ are minimal triangulations of $\pig{\pc}$ with $\pc'$ as its displayed character set. 
Then $H = H'$.
\end{lemma}
\begin{proof}
Let $\xt{}$ be an $X-$tree induced by a clique tree $\tr{}$ of $H$ and $\xt{'}$ be an $X-$tree induced by a clique tree $\tr{'}$ of $H'$.
By Lemma \ref{lem_triangulation_cliquetree_displays_legalcharacters}, $\tr{}$ and $\tr{'}$ are perfect phylogenies for $\pc'$, and because $\pc'$ defines an $X-$tree it must be that $\xt{} \isomorphic \xt{'}$ via isomorphism $\psi$.
Additionally, for each vertex $(A,\chi)$ of $\pig{\pc}$ we have $\tr{}(A,\chi) = \xt{}(A)$ and $\tr{'}(A,\chi) = \xt{'}(A)$ by Theorem \ref{thm_mintri_induced_xt}.

To prove that $H = H'$, it suffices to show that their fill edge sets are the same.
Suppose that $(A_1,\chi_1)(A_2,\chi_2)$ is a fill edge of $H$.
By the edge coverage property of clique trees, $\tr{}(A_1,\chi_1)$ and $\tr{}(A_2,\chi_2)$ intersect at a node $v$ of $T$.
We will show that $\tr{'}(A_1,\chi_1)$ and $\tr{'}(A_2,\chi_2)$ intersect at $\psi(v)$.
If there is an $a \in A_1$ such that $\phi(a) = v$, then $\psi(v) = \phi'(a)$ is a node of $\xt{'}(A_1) = \tr{'}(A_1,\chi_1)$.
Otherwise there are $a_1, a_2 \in A_1$ and $v$ is an internal node on the path from $\phi(a_1) = v_1$ to $\phi(a_2) = v_2$.
Because $\psi$ is a graph isomorphism, $\psi(v)$ is an internal node on the path from $\psi(v_1)$ to $\psi(v_2)$.
Further, $\psi(v_1)$ and $\psi(v_2)$ are nodes of $\xt{'}(A_1) = \tr{'}(A_1,\chi_1)$, so $\psi(v)$ is a node of $\tr{'}(A_1,\chi_1)$ as well.
In both cases $\psi(v)$ is a node of $\tr{'}(A_1,\chi_1)$, and a similar argument shows that $\psi(v)$ is a node of $\tr{'}(A_2,\chi_2)$.
Therefore $\tr{'}(A_1,\chi_1)$ and $\tr{'}(A_2,\chi_2)$ intersect at $\psi(v)$, so $(A_1,\chi_1)(A_2,\chi_2)$ is a fill edge of $H'$, and the fill edge set of $H$ is a subset of the fill edge set of $H'$.
A symmetric argument shows that the fill edge set of $H$ is a subset of the fill edge set of $H'$, so these fill edge sets must be equal, completing the proof.
\qed\end{proof}

\begin{proposition}\label{mainprop_1}
Let $\pc$ be a set of partial characters on $X$ and $\pc'$ be a maximal defining subset of $\pc$.
Then the following conditions hold:
	\begin{description}
	\item[(i)] $\pig{\pc}$ has a unique minimal triangulation $H$ that has $\pc'$ as its displayed characters, and no other minimal triangulation of $\pig{\pc}$ has at least $\pc'$ as its displayed character set;
	\item[(ii)] $H$ is a ternary ur-chordal graph with leafage $|X|$; and
	\item[(iii)] each edge of $H$'s unique clique tree is incontractable with respect to $\pc'$.
	\end{description}
Further, if $\pc'$ defines $\xt{}$, then $\pig{\pc,\xt{}} = H$.
\end{proposition}
\begin{proof}
To see that \textbf{(i)} holds, first observe that $\pc'$ is compatible by definition.
There is a minimal triangulation $H_1$ of $\pig{\pc}$ with at least $\pc'$ as its displayed character set by Lemma \ref{lem_legalcharacters_mintri}.
Because $\pc'$ is a maximal defining subset of $\pc$, there is no $\pc' \subset \pc^* \subseteq \pc$ that is compatible.
By Lemma \ref{lem_triangulation_cliquetree_displays_legalcharacters}, the displayed characters of $H_1$ are compatible, so this set must be exactly $\pc'$.
This is true of any minimal triangulation of $\pig{\pc}$ that has at least $\pc'$ as its displayed characters.
If $H_2$ is such a minimal triangulation, then $H_1 = H_2$ by Lemma \ref{lem_maxdef_uniqueness}, so there is a unique minimal triangulation of $\pig{\pc}$ that has at least $\pc'$ as its displayed characters.
We will refer to this unique minimal triangulation as $H$ in the remainder of the proof.

Now we show that \textbf{(ii)} holds.
Let $\xt{}$ be the $X-$tree induced by a clique tree of $H$.
By Lemma \ref{lem_triangulation_cliquetree_displays_legalcharacters}, $\xt{}$ is a perfect phylogeny for $\pc'$, and since $\pc'$ is a maximal defining subset it must be that $\pc'$ defines $\xt{}$.
Recall that $\xt{}$ is free, ternary, and distinguished by $\pc'$ according to Theorem \ref{mainthm_semple_steel}.
By Theorem \ref{thm_almostdefine_urchordal}, $\pig{\pc,\xt{}}$ is ur-chordal.
On the other hand, $H = \pig{\pc,\xt{}}$ by Theorem \ref{thm_mintri_induced_xt}, so $H$ is ur-chordal as well.
By the same theorem, $H$'s unique clique tree has the same underlying tree as $\xt{}$.
Since $\xt{}$ is ternary, this clique tree must also be ternary, so $H$ is a ternary ur-chordal graph.
This proves statement \textbf{(ii)}.

Now consider condition \textbf{(iii)}, and let $uv$ be an edge of $T$.
By Theorem \ref{mainthm_semple_steel} the edge $uv$ is distinguished by $\pc'$, so there is a character $\chi \in \pc'$ that has cells $A \neq A'$ and $u$ is a node of $\xt{}(A)$ and $v$ is a node of $\xt{}(A')$.
From Theorem \ref{thm_mintri_induced_xt} we see that $\xt{}(A) = \tr{}(A,\chi)$ and $\xt{}(A') = \tr{}(A',\chi)$, so $uv$ is incontractable with respect to $\pc'$.
Hence $\tr{}$ is incontractable with respect to $\pc'$.

The remainder of the theorem was shown while proving \textbf{(ii)} holds.
\qed\end{proof}

\noindent\textit{Proof of Theorem \ref{mainthm_subset_urchordal_defining}.}
Propositions \ref{mainprop_2} and \ref{mainprop_1} show that $\pc'$ is a maximal defining subset of $\pc$ if and only if conditions \textbf{(a)} -- \textbf{(c)} hold.
The fact that $\xt{}$ is free, ternary, and distinguished by $\pc'$ follows by Theorem \ref{mainthm_semple_steel}.
Finally, $\pig{\pc,\xt{}} = H$ due to Proposition \ref{mainprop_1}.
\qed

\vspace{\baselineskip}
\noindent\textit{Proof of Theorem \ref{mainthm_urchordal_defining}.}
Use Theorem \ref{mainthm_subset_urchordal_defining}, taking $\pc' = \pc$.
\qed

\section{Discussion}

We conclude with a brief discussion on the role minimal separators play in minimal triangulation theory \cite{H06}, and how our characterization may contribute towards constructing an algorithm that sometimes finds a maximal defining subset of characters when one exists.
Minimal triangulations have been characterized by their minimal separators, which happen to be minimal separators of the triangulated graph as well \cite{PS97,KKS97}.
Further, a minimal separator of a minimal triangulation has connected components (and full components) that are identical in the graph that has been triangulated \cite{H06}.

Bouchitt{\'e} and Todinca \cite{BT01,BT02} used minimal separators and \emph{potential maximal cliques}, the maximal cliques of minimal triangulations, to create a dynamic programming algorithm to solve the \emph{treewidth} and \emph{minimum-fill} problems in time polynomial in the number of minimal separators of a graph.
This approach was extended to create a dynamic programming algorithm that solves a variety perfect phylogeny problems in \cite{G13}, including the unique perfect phylogeny problem.

Our results elucidate the structure of minimal separators of triangulations associated with maximal defining subsets of characters.
This structure is retained in the partition intersection graph, and is closely related to the structure of potential maximal cliques, because the connected components obtained by removing the vertices in a potential maximal clique have neighborhoods that are minimal separators \cite{BT01}.
This may allow for the computation of a ternary ur-chordal minimal triangulation in time polynomial in the number of minimal separators of $\pig{\pc}$ (or asserting that no ternary ur-chordal minimal triangulations exist), yielding a candidate subset $\pc'$ of $\pc$ that may be a maximal subset of characters.
The number of minimal separators of $\pig{\pc'}$ is bounded by the number of minimal separators of $\pig{\pc}$ (this is a specific example of a more general fact; see Corollary 4 in \cite{BT02}).
Therefore if it is computationally feasible to find $\pc'$ due to $\pig{\pc}$ having a small number of minimal separators, checking if $\pc'$ defines a perfect phylogeny using the method from \cite{G13} may also be feasible.

\section{Acknowledgements}
This research was partially supported by NSF grants IIS-0803564 and CCF-1017580.

\bibliographystyle{plain}
\bibliography{../master}

\end{document}